\documentclass{elsarticle}

\usepackage{amsmath,amssymb}
\usepackage{amsthm}
\usepackage{hyperref}

\newtheorem{thm}{Theorem}
\newtheorem{lem}{Lemma}

\newtheorem{rk}{Remark}
\newtheorem{ex}{Example}

\begin{document}

\begin{frontmatter}

\title{Two classes of linear codes and their generalized Hamming weights}

\author[mymainaddress]{Gaopeng Jian\corref{mycorrespondingauthor}}
\cortext[mycorrespondingauthor]{Corresponding author}   
\ead{gpjian@pku.edu.cn}
   
\author[mymainaddress]{Zhouchen lin}  
\ead{zlin@pku.edu.cn}
  
\author[mysecondaryaddress]{Rongquan Feng}   
\ead{fengrq@math.pku.edu.cn}  
  
\address[mymainaddress]{Key Laboratory of Machine Perception(MOE), School of EECS, Peking University, Beijing 100871, P.R.China}  
\address[mysecondaryaddress]{LMAM, School of Mathematical Sciences, Peking University, Beijing 100871, P.R.China}  

\begin{abstract}
The generalized Hamming weights (GHWs) are fundamental parameters of linear codes. 
In this paper, we investigate the generalized Hamming weights of two classes of linear codes constructed from defining sets and determine them completely employing a number-theoretic approach.
\end{abstract}


\end{frontmatter}

\section{Introduction}
Let $\mathbb{F}_p$ be the finite field with $p$ elements, where $p$ is an odd prime.
An $[n,k]$ linear code $C$ over $\mathbb{F}_p$ is a $k$-dimensional subspace of the linear
space $\mathbb{F}_p^n$. 
For any linear subcode $U \subset C$, the support of $U$ is defined to be
\[
\text{Supp}(U) = \{i : 1 \le i \le n, \ x_i \neq 0 \text{ for some } x=(x_1,x_2,\ldots,x_n) \in U \}.
\]
For $1 \le r \le k$, the $r$-th generalized Hamming weight (GHW) of $C$ is given by
\[
d_r (C)=\min \{|\text{Supp}(U)|: U \subset C, \  \text{dim}_{\mathbb{F}_p}(U)=r\},
\]
where $|\text{Supp}(U)|$ denotes the cardinality of  $\text{Supp}(U)$. 
By definition, $d_1(C)$ is just the minimum distance of $C$. 
The set $\{d_r (C) : 1 \le r \le k\}$ is called the weight hierarchy of $C$.

The concept of GHWs was first introduced in \cite{helleseth1977weight,klove1978weight}, and rediscovered by Wei \cite{wei1991generalized} to fully characterize the performance of linear codes when used in a wire-tap channel of type II or as a t-resilient function. 
Indeed, the GHWs provide detailed structural information of linear codes, which can also be used to compute the state complexity of trellis diagrams for linear block codes \cite{forney1994dimension}, to determine the erasure list-decodability of linear codes \cite{guruswami2003list} and so on.

In general, the determination of weight hierarchy is very difficult and there are only a few classes of linear codes whose weight hierarchies are known (see Section \ref{clu}). 
In this paper, we construct two classes of linear codes and determine their weight hierarchies using exponential sums. In some cases they are optimal with respect to known bounds.

The rest of this paper is organized as follows. 
In Section \ref{pl}, we review basic concepts and results on exponential sums together with previously known results on GHWs.
In Section \ref{wt}, we investigate the generalized Hamming weights of two classes of linear codes. 
In Section \ref{clu}, we conclude this paper.

\section{Preliminaries} \label{pl}
\subsection{Characters, cyclotomic classes and exponential sums over finite fields} 
We introduce several basic results on characters, cyclotomic classes and exponential sums used in this paper. 
For more details, please see Chapter 5 of the book \cite{lidl1997finite}.

Let $\mathbb{F}_q$ be the finite field with $q$ elements, where $q$ is a power of a prime $p$.
Define the \emph{canonical additive character} of $\mathbb{F}_q$ as
\[
\chi: \mathbb{F}_q \longrightarrow \mathbb{C}^*, \ \chi(x) = \zeta_p^{\text{Tr}_{q/p}(x)},
\]
where $\zeta_p=e^{\frac{2 \pi \sqrt{-1}}{p}}$ is the primitive $p$-th root of unity and $\text{Tr}_{q/p}$ is the trace function from $\mathbb{F}_q$ to $\mathbb{F}_p$.
The orthogonal property of additive characters is given by
\[
\sum_{x \in \mathbb{F}_q} \chi(ax)=\begin{cases}
q, & \text{if } a=0, \\
0, & \text{if } a \in \mathbb{F}_q^*.
\end{cases}
\]

For $\mathbb{F}^*_q=\langle \alpha \rangle$ and a positive integer $N>1$ such that $N|(q-1)$, we define
\[
C_i^{(N,q)}=\alpha^i \langle \alpha^N \rangle, \ i=0,1,\ldots,N-1,
\]
where $\langle \alpha^N \rangle$ denotes the cyclic subgroup of $\mathbb{F}_q^*$ generated by $\alpha^N$.
The cosets $C_i^{(N,q)}$ are called the \emph{cyclotomic classes} of order $N$ in  $\mathbb{F}_q$. 

The \emph{Gaussian periods} of order $N$ over $\mathbb{F}_q$ are defined by
\[
\eta_i^{(N,q)}=\sum_{x \in C_i^{(N,q)}} \chi(x),\ i=0,1,\ldots,N-1.
\]

\begin{lem}  \label{l1}
Suppose that $q=p^m$ and $N=2$. The Gaussian periods are given by
\[
\eta_0^{(2,q)}=\begin{cases}
\frac{-1+(-1)^{m-1}\sqrt{q}}{2}, & \text{if} \ p \equiv 1 \pmod{4},\\
\frac{-1+(-1)^{m-1}(\sqrt{-1})^{m} \sqrt{q}}{2}, & \text{if} \ p \equiv 3 \pmod{4}
\end{cases}
\]
and
\[
\eta_1^{(2,q)}=-1-\eta_0^{(2,q)}.
\]
\end{lem}

Let $M \ge 2$ be an integer and $q=p^{2fh}$, where $f$ is the least positive integer satisfying $p^f \equiv -1 \pmod M$ and $\text{gcd}(h,p)=1$. 
The values of the following exponential sum
\[
\Omega(a,b)=\sum_{x \in \mathbb{F}_q^*}\chi \left( ax^{\frac{q-1}{M}}+bx \right), \ a \in \mathbb{F}_q^*, \ b \in \mathbb{F}_q
\] 
were determined in terms of Gauss periods in \cite{li2015walsh}.
We summarize the results in the following lemma.

\begin{lem}\label{l2}
Let $a \in C_t^{(d,q)}$ for some $t$, where $d=\frac{q-1}{M}$.
\begin{enumerate}
\item If b=0, then
\[
\Omega(a,0)=d\eta_t^{(d,q)}.
\]
\item If $b \in C_{\delta}^{(M,q)}$, $0 \le \delta \le M-1$, then
\[
\Omega(a,b)=\sqrt{q} \chi(-a \alpha^{-d \delta})-\frac{(\sqrt{q}+1)\eta_t^{(d,q)}}{M}
\]
if $p$, $h$ and $\frac{p^f+1}{M}$ are odd and
\[
\Omega(a,b)=(-1)^{h-1}\sqrt{q} \chi(a \alpha^{-d \delta})+\frac{((-1)^h\sqrt{q}-1)\eta_t^{(d,q)}}{M}
\] 
otherwise.
\end{enumerate}
\end{lem}

\subsection{Bounds and formulas of GHWs}
Here we present three bounds on GHWs of linear codes. The reader may refer to the literature \cite{tsfasman1995geometric} for them.
\begin{lem}
Let $C$ be an $[n,k]$ linear code over $\mathbb{F}_p$. 
For $1 \le r \le k$,
\begin{description}
\item[(1) Singleton type bound:]
\[
r \le d_r(C) \le n-k+r.
\]
$C$ is called an $r$-MDS code if $d_r(C)=n-k+r$.
\item[(2) Plotkin-like bound:]
\[
d_r(C) \le \left \lfloor \frac{n(p^r-1)p^{k-r}}{p^k-1} \right \rfloor.
\] 
\item[(3) Griesmer-like bound:]
\[
d_r(C) \ge \sum_{i=0}^{r-1} \left \lceil \frac{d_1(C)}{p^i} \right \rceil.
\]
\end{description}
\end{lem}

We recall a generic construction of linear codes proposed by Ding et al. \cite{ding2007cyclotomic}.
Let $D=\{d_1,d_2,\ldots,d_n \}\subset \mathbb{F}_q$, define a $p$-ary linear code of length $n$ by
\[
C_D=\{(\text{Tr}_{q/p}(a d_1), \text{Tr}_{q/p}(a d_2), \ldots, \text{Tr}_{q/p}(a d_n)): a \in \mathbb{F}_q\}.
\]
The set $D$ is called the \emph{defining set} of $C_D$.
Recently, many good linear codes were constructed by chosing appropriate defining sets \cite{tang2016linear,zhou2015linear,heng2017construction,ding2014binary,ding2015class,xu2018complete}. 

To calculate GHWs of the code $C_D$, we present two formulas which are essentially proved in \cite{yang2015generalized} (see also \cite{jian2017generalized,li2018class}).
For convenience we denote by $[L,r]_p$ the set of $r$-dimensional subspaces of $L$ for any $\mathbb{F}_p$-linear space $L$ and $H^*=H \backslash \{0\}$ for any $H \in [L,r]_p$.

\begin{lem} \label{t1}
Suppose that $q=p^m$ and $\text{dim}_{\mathbb{F}_p}(C_D)=m$.
For $1 \le r \le m$, 
\[
d_r(C_D)=n-\max\{ \ |D \cap H| : H \in [\mathbb{F}_q,m-r]_p \ \}.
\]
\end{lem}

\begin{lem} \label{t2}
Suppose that $q=p^m$ and $\text{dim}_{\mathbb{F}_p}(C_D)=m$.
For $1 \le r \le m$, 
\[
d_r(C_D)=n-\max\{ \ N(H_r) : H_r \in [\mathbb{F}_q,r]_p \ \},
\]
where
\[
N(H_r)=\frac{n}{p^r}+\frac{1}{p^r} \sum_{a \in H_r^*} \sum_{x \in D} \chi (a x). 
\]
\end{lem}

\section{Main results and proofs} \label{wt}
From now on we fix the following notations.
\begin{itemize}
\item $q=p^m$, $p$ is an odd prime, $m$ is a positive integer with $\text{gcd}(m,p)=1$.
\item $\alpha$ is a primitive element of $\mathbb{F}_q$.
\item $d=1$ or $\frac{q-1}{p+1}$.
\item $s=\frac{m}{2}$, we only consider the case that $s$ is an odd integer if $d=\frac{q-1}{p+1}$.
\item $\zeta_p=e^{\frac{2 \pi \sqrt{-1}}{p}}$ is the primitive $p$-th root of unity. 
\item $\chi_h$ is the canonical additive character of $\mathbb{F}_{p^h}$. 
\item $\text{Tr}_e^h$ is the trace function from $\mathbb{F}_{p^h}$ to $\mathbb{F}_{p^e}$.
\item The defining set $D$ of the code $C_D$ is
\[
D=\{ x \in \mathbb{F}_q^*: \text{Tr}^m_1(x^d)=0 \}.
\]
\item $\widetilde{D}=D \cup \{0\}$.
\end{itemize}

Let $c$ be a mapping from $\mathbb{F}_q$ to $\mathbb{F}_p^n$ defined by
\[
c(a)=(\text{Tr}^m_1(a x))_{x \in D}, 
\]
for each $a \in \mathbb{F}_q$. 
Obviously, $c$ is $\mathbb{F}_p$-linear and the image of $c$ is $C_D$.
If $d=\frac{q-1}{p+1}$, we shall see that $c$ is injective and thus induces a 1-1 correspondence between $[\mathbb{F}_q,r]_p$ and $[C_D,r]_p$ for $1 \le r \le m$.
If $d=1$, however $c$ is not injective and the kernel of $c$ is $\mathbb{F}_p$.
It's easy to see that $\widetilde{D}$ is a $(m-1)$-dimensional subspace of $\mathbb{F}_q$ and $\mathbb{F}_q=\widetilde{D} \oplus \mathbb{F}_p$. 
Then $c$ induces a 1-1 correspondence between $[\widetilde{D},r]_p$ and $[C_D,r]_p$ for $1 \le r \le m-1$.

By the orthogonal property of additive characters we can determine the length and dimension of $C_D$.
\begin{thm}
Let $n=|D|$ and $k=\text{dim}_{\mathbb{F}_p}(C_D)$, then
\[
n=\begin{cases}
p^{m-1}-1, & \text{if } d=1, \\
0, & \text{if $d=\frac{q-1}{p+1}$ and $p \equiv 1 \pmod{4}$}, \\
\frac{2(q-1)}{p+1} & \text{if $d=\frac{q-1}{p+1}$ and $p \equiv 3 \pmod{4}$},
\end{cases}
\]
and
\[
k=\begin{cases}
m-1, & \text{if } d=1, \\
0, & \text{if $d=\frac{q-1}{p+1}$ and $p \equiv 1 \pmod{4}$}, \\
m, & \text{if $d=\frac{q-1}{p+1}$ and $p \equiv 3 \pmod{4}$}.
\end{cases}
\]
\end{thm}

\begin{proof}
By definition,
\begin{align*}
n &=\sum_{x \in \mathbb{F}_q^*}  \left( \frac{1}{p}\sum_{y \in \mathbb{F}_p} \zeta_p^{y\text{Tr}^m_1(x^d)} \right) \\
&=\frac{1}{p} \sum_{y \in \mathbb{F}_p} \sum_{x \in \mathbb{F}_q^*} \chi_m(yx^d)
\end{align*}
and for $a \in \mathbb{F}_q$, the Hamming weight of $c(a)$ is
\begin{align*}
W_H(c(a)) &=n-\sum_{x \in D} \left( \frac{1}{p}\sum_{y \in \mathbb{F}_p} \zeta_p^{y\text{Tr}^m_1(ax)} \right) \\
&=n-\frac{1}{p} \sum_{y \in \mathbb{F}_p} \sum_{x \in D} \chi_m(yax) \\
&=n-\frac{1}{p} \sum_{y \in \mathbb{F}_p} \sum_{x \in \mathbb{F}_q^*}  \left( \frac{1}{p}\sum_{z \in \mathbb{F}_p} \zeta_p^{z\text{Tr}^m_1(x^d)} \right) \chi_m(yax) \\
&=n-\frac{1}{p^2} \sum_{y \in \mathbb{F}_p} \sum_{z \in \mathbb{F}_p} \sum_{x \in \mathbb{F}_q^*}\chi_m(zx^d+yax). 
\end{align*}

\begin{enumerate}
\item If $d=1$,
\begin{align*}
n &=\frac{1}{p} \sum_{y \in \mathbb{F}_p} \sum_{x \in \mathbb{F}_q^*} \chi_m(yx) \\
&=\frac{1}{p}(q-1+(p-1)(-1)) \\
&=p^{m-1}-1.
\end{align*}

If $a \notin \mathbb{F}_p$, 
\begin{align*}
W_H(c(a)) &=n-\frac{1}{p^2} \sum_{y \in \mathbb{F}_p} \sum_{z \in \mathbb{F}_p} \sum_{x \in \mathbb{F}_q^*}\chi_m((z+ya)x) \\
&=n-\frac{1}{p^2}(q-1+(p^2-1)(-1)) \\
&=p^{m-1}-p^{m-2}
\end{align*}
and $c(a)$ is the zero vector if $a \in \mathbb{F}_p$.
So $k=m-1$.

\item If $d=\frac{q-1}{p+1}$, by Lemma \ref{l2},
\begin{align}
n &=\frac{q-1}{p}+\frac{1}{p}\sum_{y \in \mathbb{F}_p^*}\Omega(y,0) \notag \\
&=\frac{q-1}{p}+\frac{q-1}{p(p+1)}\sum_{y \in \mathbb{F}_p^*}\eta_{t(y)}^{(d,q)}, \label{n1}
\end{align}
where $y \in C_{t(y)}^{(d,q)}$ for $y \in  \mathbb{F}_p^*$.
Let $\beta=\alpha^{\frac{q-1}{p^2-1}}$, which is a primitive element of $\mathbb{F}_{p^2}$, then
\begin{align*}
\eta_{t(y)}^{(d,q)} &=\sum_{i=0}^p \chi_m \left( y\alpha^{\frac{(q-1)i}{p+1}}\right) \\
&=\sum_{i=0}^p \chi_m ( y\beta^{(p-1)i}) \\
&=\sum_{i=0}^p \chi_2 ( \text{Tr}^m_2(y\beta^{(p-1)i})) \\
&=\sum_{i=0}^p \chi_2 ( sy\beta^{(p-1)i}).
\end{align*}
Since $\text{gcd}(m,p)=1$,  $s \in \mathbb{F}_p^*$.
Let $I=\{f: \text{$f$ is an even number with } 0 \le f \le p-2\}$. 
Note that $\beta^{p+1}$ is a primitive element of $\mathbb{F}_p$, then 
\begin{align}
\sum_{y \in \mathbb{F}_p^*}\eta_{t(y)}^{(d,q)} &=\sum_{y \in \mathbb{F}_p^*} \sum_{i=0}^p \chi_2 (y\beta^{(p-1)i}) \notag \\
&=\sum_{j=0}^{p-2}\sum_{i=0}^p \chi_2 (\beta^{(p+1)j+(p-1)i}) \notag \\
&=\sum_{j=0}^{\frac{p-3}{2}}\left(\sum_{i=0}^p \chi_2 ( \beta^{2j+(p-1)(i+j)})+\sum_{i=0}^p \chi_2 ( \beta^{-2j+(p-1)(i-j)})\right)  \notag \\
&=\sum_{j \in I} \left(\sum_{i=0}^p \chi_2 ( \beta^{j+(p-1)i})+\sum_{i=0}^p \chi_2 ( \beta^{-j+(p-1)i})\right)  \notag \\
&=\sum_{j \in I} \left(\sum_{i=0}^p \chi_2 ( \beta^{j+(p-1)i})+\sum_{i=0}^p \chi_2 ( \beta^{-j+(p-1)i})\right) \notag \\
&=2 \sum_{j \in I} \eta_j^{(p-1,p^2)}=2\eta_0^{(2,p^2)}. \label{n2}
\end{align}
By Lemma \ref{l1} and formula \eqref{n1}, \eqref{n2}  
\[
n=\begin{cases}
0, & \text{if } p \equiv 1 \pmod{4}, \\
\frac{2(q-1)}{p+1}, & \text{if } p \equiv 3 \pmod{4}.
\end{cases}
\]
So if $p \equiv 1 \pmod{4}$, the defining set $D$ is empty and $k=0$. 
If $p \equiv 3 \pmod{4}$, let $a \in C_{\delta}^{(p+1,q)}$, $0 \le \delta \le p$. 
We have $\mathbb{F}_p^* \subset C_0^{(p+1,q)}$ noting that $\alpha^{\frac{q-1}{p-1}}$ is a primitive element of $\mathbb{F}_p$ and $p+1|\frac{q-1}{p-1}$.
By Lemma \ref{l1}, \ref{l2} and formula \eqref{n2}
\begin{align*}
W_H(c(a))= &n-\frac{q-p}{p^2}-\frac{1}{p^2} \sum_{z \in \mathbb{F}_p^*}\Omega(z,0)-\frac{1}{p^2} \sum_{y \in \mathbb{F}_p^*} \sum_{z \in \mathbb{F}_p^*}\Omega(z,ya) \\
= &n-\frac{q-p}{p^2}-\frac{q-1}{p^2(p+1)}\sum_{z \in \mathbb{F}_p^*}\eta_{t(z)}^{(d,q)} \\
 &-\frac{1}{p^2} \sum_{y \in \mathbb{F}_p^*} \sum_{z \in \mathbb{F}_p^*}\left( p^s \chi_m(-z \alpha^{-d \delta})-\frac{p^s+1}{p+1}\eta_{t(z)}^{(d,q)} \right) \\
=&\frac{2(q-1)}{p+1}-\frac{q-p}{p^2}+\frac{p^s(p-1)+p-q}{p^2(p+1)}\sum_{z \in \mathbb{F}_p^*}\eta_{t(z)}^{(d,q)} \\
&-(p-1)p^{s-2} \sum_{z \in \mathbb{F}_p^*} \zeta_p^{-z \text{Tr}^m_1(\alpha^{-d \delta})} \\
=& \begin{cases}
\frac{p^{s-1}(2p^s-p+1)(p-1)}{p+1}, & \text{if } \text{Tr}^m_1(\alpha^{-d \delta})=0, \\
\frac{2p^{s-1}(p^s+1)(p-1)}{p+1}, & \text{if } \text{Tr}^m_1(\alpha^{-d \delta}) \ne 0,
\end{cases}
\end{align*}
which implies that $k=m$.
\end{enumerate}
\end{proof}

Now we determine GHWs of the code $C_D$.
\begin{thm}
If $d=1$, then
\begin{equation} \label{ghw1}
d_r(C_D)=p^{m-1}\left(1-\frac{1}{p^r}\right), \ 1 \le r \le m-1.
\end{equation}   
\end{thm}

\begin{proof}
Indeed we can not use Lemma \ref{t1} directly, but by the same principle we have 
\begin{align*}
d_r(C_D) &=n-\max\{ \ |D \cap H| : H \in [\widetilde{D},m-1-r]_p \ \} \\
&=p^{m-1}-1-(p^{m-1-r}-1) \\
&=p^{m-1}\left(1-\frac{1}{p^r}\right)
\end{align*} 
for $1 \le r \le m-1$.
\end{proof}

\begin{rk}
By formula \eqref{ghw1} it's easy to check that $C_D$ is $(m-1)$-MDS and $d_r(C_D)$ meets the Plotkin-like bound and the Griesmer-like bound for all $1 \le r \le m-1$.
\end{rk}

\begin{ex}
Let $(p,m)=(3,3)$ and $d=1$, then $C_D$ is a $[8,2]$ linear code over $\mathbb{F}_3$ with $d_1(C_D)=6$ and  $d_2(C_D)=8$
$C_D$ is 2-MDS. 
$d_r(C_D)$ meets the Plotkin-like bound and the Griesmer-like bound for all $1 \le r \le 2$.
\end{ex}

\begin{ex}
Let $(p,m)=(3,6)$ and $d=1$, then $C_D$ is a $[242,5]$ linear code over $\mathbb{F}_3$ with $d_1(C_D)=162$, $d_2(C_D)=216$, $d_3(C_D)=234$, $d_4(C_D)=240$ and $d_5(C_D)=242$.
$C_D$ is 5-MDS. 
$d_r(C_D)$ meets the Plotkin-like bound and the Griesmer-like bound for all $1 \le r \le 5$.
\end{ex}

\begin{thm}
If $d=\frac{q-1}{p+1}$ and $ p \equiv 3 \pmod{4}$, then
\begin{equation} \label{ghw2}
d_r(C_D)=\begin{cases}
\frac{p^s(2p^s+1-p)}{p+1}\left(1-\frac{1}{p^r}\right), & \text{if } 1 \le r \le s, \\
\frac{2(q-1)}{p+1}+1-p^{m-r}, & \text{if } s \le r \le m.
\end{cases}
\end{equation}   
\end{thm}

\begin{proof}
With the symbols defined in Lemma \ref{t2}. 
For $H_r \in [\mathbb{F}_q,r]_p$, let \[
H_r^0=H_r \bigcap \left(\bigcup_{\mbox{\tiny $\begin{array}{c}
0 \le \delta \le p, \\
\text{Tr}^m_1(\alpha^{-d \delta})=0 \end{array}$}}C_{\delta}^{(p+1,q)}\right).
\] 
By Lemma \ref{l1}, \ref{l2} and formula \eqref{n2}
\begin{align*}
N(H_r) &=\frac{n}{p^r}+\frac{1}{p^r} \sum_{a \in H_r^*} \sum_{x \in \mathbb{F}_q^*} \left( \frac{1}{p}\sum_{y \in \mathbb{F}_p} \zeta_p^{y\text{Tr}^m_1(x^d)} \right) \chi_m (a x) \\
&=\frac{n}{p^r}+\frac{1}{p^{r+1}}\sum_{a \in H_r^*}\sum_{y \in \mathbb{F}_p}\sum_{x \in \mathbb{F}_q^*} \chi_m(yx^d+a x)  \\ 
&=\frac{pn+1-p^r}{p^{r+1}}+\frac{1}{p^{r+1}}\sum_{a \in H_r^*}\sum_{y \in \mathbb{F}_p^*}\Omega(y,a). \\
&=\frac{pn+1-p^r}{p^{r+1}}+\frac{1}{p^{r+1}}\sum_{a \in H_r^*}\sum_{y \in \mathbb{F}_p^*}\left( p^s \chi_m(-y \alpha^{-d \delta})-\frac{p^s+1}{p+1}\eta_{t(y)}^{(d,q)} \right)  \\
&=\frac{pn+1-p^r}{p^{r+1}}-\frac{(p^r-1)(p^s+1)}{p^{r+1}(p+1)}\sum_{y \in \mathbb{F}_p^*}\eta_{t(y)}^{(d,q)}+p^{s-r-1} \sum_{a \in H_r^*}\sum_{y \in \mathbb{F}_p^*} \zeta_p^{-y \text{Tr}^m_1(\alpha^{-d \delta})} \\
&=\frac{pn+1-p^r}{p^{r+1}}-\frac{(p-1)(p^r-1)(p^s+1)}{p^{r+1}(p+1)}+p^{s-r-1} \left( (p-1)|H_r^0|-|H_r^* \backslash H_r^0| \right).
\end{align*}
Let $\delta_0=\frac{p+1}{4}$ and $\beta=\alpha^{\frac{q-1}{p^2-1}}$, then $\text{Tr}^m_1(\alpha^{-d \delta_0})=\text{Tr}^m_1(\beta^{-\frac{p^2-1}{4}})=s\text{Tr}^2_1(\beta^{-\frac{p^2-1}{4}})=s\beta^{-\frac{p^2-1}{4}}(1+\beta^{-\frac{(p^2-1)(p-1)}{4}})=s\beta^{-\frac{p^2-1}{4}}(1+\beta^{\frac{(p^2-1)}{2}-\delta_0(p^2-1)})=0$.

If $1 \le r \le s$, note that $\mathbb{F}_{p^s}^* \subset C_0^{(p+1,q)}$ since $\alpha^{\frac{q-1}{p^s-1}}$ is a primitive element of $\mathbb{F}_{p^s}$ and $p+1|\frac{q-1}{p^s-1}$.
We can choose $U_1 \in [x_1\mathbb{F}_{p^s},r]_p$, where $x_1 \in C_{\delta_0}^{(p+1,q)}$, and  $\max\{ \ N(H_r) : H_r \in [\mathbb{F}_q,r]_p \ \}=N(U_1)=\frac{pn+1-p^r}{p^{r+1}}-\frac{(p-1)(p^r-1)(p^s+1)}{p^{r+1}(p+1)}+p^{s-r-1}(p-1)(p^r-1)$. 
By Lemma \ref{t2}, 
\[
d_r(C_D)=\frac{p^s(2p^s+1-p)}{p+1}\left(1-\frac{1}{p^r}\right).
\]

If $s \le r \le m$, note that $a^d=\left(a^{p^s-1}\right)^{\frac{p^s+1}{p+1}}=1$ for any $a \in \mathbb{F}_{p^s}^*$, which implies that $x_2 \mathbb{F}_{p^s} \in [\mathbb{F}_q,s]_p$ is contained in $\widetilde{D}$ for any $x_2 \in D$.
We can choose $U_2 \in [x_2\mathbb{F}_{p^s},m-r]_p$ and $\max \{ \ |D \cap H| : H \in [\mathbb{F}_q,m-r]_p \ \}=|U_2^*|=p^{m-r}-1$. 
By Lemma \ref{t1}, 
\[
d_r(C_D)=\frac{2(q-1)}{p+1}+1-p^{m-r}.
\]
\end{proof}

\begin{rk}
By formula \eqref{ghw2} it's easy to check that $C_D$ is $m$-MDS, $d_r(C_D)$ meets the Griesmer-like bound for $1 \le r \le s$ and $d_m(C_D)$ meets the Plotkin-like bound.
\end{rk}

\begin{ex}
Let $(p,m)=(3,2)$ and $d=\frac{p^m-1}{p+1}=2$, then $C_D$ is a $[4,2]$ linear code over $\mathbb{F}_3$ with $d_1(C_D)=2$ and $d_2(C_D)=4$.
$C_D$ is 2-MDS.
$d_2(C_D)$ meets the Plotkin-like bound.
\end{ex}

\begin{ex}
Let $(p,m)=(3,6)$ and $d=\frac{p^m-1}{p+1}=182$, then $C_D$ is a $[364,6]$ linear code over $\mathbb{F}_3$ with $d_1(C_D)=234$, $d_2(C_D)=312$, $d_3(C_D)=338$, $d_4(C_D)=356$, $d_5(C_D)=362$ and $d_6(C_D)=364$.
$C_D$ is 6-MDS.
$d_r(C_D)$ meets the Griesmer-like bound for $1 \le r \le 3$. 
$d_6(C_D)$ meets the Plotkin-like bound.
\end{ex}

\section{Concluding remarks} \label{clu}
The weight hierarchy of a code has been examined at least in the following cases:
\begin{enumerate}
\item Hamming codes
\item Golay codes
\item Product codes.
\item Codes from classical varieties: Reed–Muller codes, Algebraic geometric codes, codes from quadrics, Hermitian varieties, Grassmannians, Del Pezzo surfaces. 
\item Binary Kasami codes.
\item Cyclic and trace codes: BCH, Melas.
\item Codes parameterized by the edges of simple graphs.
\end{enumerate}
A survey up to known results until 1995 was done in \cite{tsfasman1995geometric}. 
Recent results can be found in \cite{yang2015generalized,ballico2016higher,xiong2016weight,sarabia2016generalized,hao2017weight,li2017weight,jian2017generalized,johnsen2017generalized,beelen2018note,sarabia2018second,beelen2018generalized,li2018class,liu2018some}.

\section*{Acknowledgments}
G.Jian and Z.Lin are supported by National Natural Science Foundation (NSF) of China (grant no.61625301).
R.Feng is supported by the NSFC-Genertec Joint Fund For Basic Research (grant no.U1636104).

\section*{References}
\bibliographystyle{elsarticle-harv}
\bibliography{thesis} 
\end{document}